\documentclass[a4paper,UKenglish,cleveref, autoref, thm-restate,  numberwithinsect]{lipics-v2021}
\usepackage{mathtools, tikz}
\usetikzlibrary{arrows.meta,positioning}
\nolinenumbers

\newcommand{\F}{\mathbb{F}}

\newcommand{\N}{\mathbb{N}}

\newcommand{\Q}{\mathbb{Q}}

\newcommand{\bfx}{\mathbf{x}}
\newcommand{\bfy}{\mathbf{y}}
\newcommand{\bfa}{\mathbf{a}}

\newcommand{\VP}{\mathsf{VP}}
\newcommand{\VNP}{\mathsf{VNP}}
\newcommand{\sharpP}{\mathsf{\#P}}

\newcommand{\Sym}{\mathsf{Sym}}

\newcommand{\gen}{\mathsf{Gen}}

\newcommand{\alg}{\mathcal{A}}

\newcommand{\sdet}{\mathsf{sdet}}

\newcommand{\sgn}{\operatorname{sgn}}
\newcommand{\mat}{\operatorname{Mat}}
\newcommand{\tr}{\operatorname{Tr}}
\newcommand{\Cdet}{\mathsf{Cdet}}

\newcommand{\vnpac}{\mathsf{VNP}_{\mathsf{A,\overline{C}}}}
\newcommand{\vpac}{\mathsf{VP}_{\mathsf{A,\overline{C}}}}
\newcommand{\hcac}{\mathsf{HC}_{\mathsf{A,\overline{C}}}}

\def\final{0}  
\def\iflong{\iffalse}
\ifnum\final=0  
\definecolor{lightskyblue}{rgb}{0.53, 0.81, 0.98}\newcommand{\sanyam}[1]{{\small{\color{lightskyblue}[Sanyam: #1]}}}
\newcommand{\markus}[1]{{\small{\color{blue}[Markus: #1]}}}
\newcommand{\mridul}[1]{{\small{\color{green}[Mridul: #1]}}}

\else 
\newcommand{\sanyam}[1]{}
\newcommand{\markus}[1]{}
\newcommand{\mridul}[1]{}
\fi  

\newcommand{\eps}{\varepsilon}

\DeclareMathSymbol{\nmid}{\mathrel}{AMSb}{"2D}

\renewcommand{\tilde}{\widetilde}

\newtheorem*{theorem*}{Theorem}

\title{On the Principal Minor Expansion and Complexity of the Symmetrized Determinant} 

\author{Sanyam Agarwal {}}{Universit\"at des Saarlandes, Saarbr\"ucken, Germany}{agarwal@cs.uni-saarland.de}{}{}
\author{Markus Bl\"aser {}}{Universit\"at des Saarlandes, Saarbr\"ucken, Germany}{mblaeser@cs.uni-saarland.de}{}{}

\author{Mridul Gupta {}}{Indian Institute of Technology Kanpur, Kanpur, India}{mridulg@cse.iitk.ac.in}{}{}

\authorrunning{Agarwal, Bl\"aser, and Gupta} 

\Copyright{Sanyam Agarwal, Markus Bl\"aser, and Mridul Gupta} 

\ccsdesc[500]{Theory of computation~Algebraic complexity theory} 

\keywords{Symmetrized Determinants, $\VP$, $\VNP$, $\sharpP$} 

\category{} 

\relatedversion{} 

\EventEditors{}
\EventNoEds{2}
\EventLongTitle{}
\EventShortTitle{}
\EventAcronym{}
\EventYear{}
\EventDate{}
\EventLocation{}
\EventLogo{}
\SeriesVolume{}
\ArticleNo{XX}

\begin{document}

\maketitle

\begin{abstract}
    Barvinok \cite{Barvinok2000NewPE} introduced the symmetrized determinant ($\sdet$) as a \emph{non-commutative} analogue of the determinant. Intuitively, given a square matrix over an associative algebra, we can obtain the symmetrized determinant by averaging over all possible multiplication orders in the Leibniz formula for the determinant. He used the symmetrized determinant to design algorithms estimating the permanent of a matrix. To this end, he showed that there is a $O(n^{r+3})$ algorithm computing $\sdet$, where $r$ is the dimension of the algebra, and is therefore polynomial-time computable for fixed $r$.

    In this work, we study the algebraic properties and complexity of $\sdet$. While most of the properties of the ordinary determinant don't generalize to $\sdet$ defined on non-commutative algebras, we show that the principal minor expansion of the $\sdet$ is analogous to the ordinary determinant. Second, we prove that there exists a polynomial-sized algebra such that computing the symmetrized determinant is $\sharpP$-hard. Third, we show that the associated polynomial family is $\VNP$-complete over a suitable polynomial-dimensional algebra in the non-commutative setting. Further, when seen as a family of polynomials over the matrix algebra, it is also $\VNP$-complete in the commutative setting. This places the symmetrized determinant among the natural complete families arising from algebraic computation.
\end{abstract}

\section{Introduction}
The determinant and the permanent occupy a central place in algebraic complexity theory \cite{burgisser2000completeness, ramprasadsurvey}. Although they are defined by almost identical summation formulas, the determinant is polynomial-time computable and in $\VP$, while the permanent is $\sharpP$ and $\VNP$ complete \cite{valiant1979completeness}. This contrast has naturally raised questions about the expressivity of determinant-like polynomials as compared to the permanent. A common approach has been to look towards non-commutative versions of the determinant and study their computability. The most standard way to introduce non-commutativity is to look at the \emph{Cayley} determinant. Formally, let $\alg$ be an associative algebra over a ring $\F$ of characteristic $0$. Let $M$ be a $n \times n$ matrix where each entry of $M$ (denoted $m_{ij}$) comes from $\alg$. We define the Cayley determinant as follows:
\begin{equation}\label{eq:cdet}
    \Cdet M= \sum_{\sigma \in \Sym_n} \sgn (\sigma) m_{1,\sigma(1) } \cdots m_{n,\sigma(n)}
\end{equation}
Over non-commutative algebras this inherently introduces an order of multiplication: we take the first element from the first row, second from second row, and so on. Along these lines, Barvinok \cite{Barvinok2000NewPE} introduced the notion of symmetrized determinants ($\sdet$). It extends the notion of determinants to non-commutative settings slightly differently compared to the Cayley determinant. Here, we take into account all orderings of multiplication for each permutation in the Leibniz formula for the standard determinant, essentially removing the ordering inherent in the Cayley determinant. The symmetrized determinant of matrix $M$, denoted $\sdet M$, is defined as the following:
\begin{equation}\label{eq:sdet}
    \sdet M={1 \over n!} \sum_{(\sigma, \tau) \in \Sym_n \times \Sym_n} (\sgn \sigma)(\sgn \tau) m_{\sigma(1) \tau(1)} \cdots m_{\sigma(n) \tau(n)}
\end{equation}
Barvinok showed that there exists a $O(n^{r+3})$ algorithm to compute $\sdet M$, where $r$ is the dimension of the algebra as an $\F$-vector space. Hence, for fixed dimension algebras (for example, constant sized matrices), computing the $\sdet$ can be achieved in polynomial time. Barvinok used these symmetrized determinants to design a randomized polynomial
time algorithm to estimate the permanent of non-negative matrices. This perspective was later sharpened by Moore and Russell \cite{moore2012approximating}, who studied permanent estimators over matrix and semisimple algebras based on these non-commutative determinants. They analyzed the variance of these estimators and showed that estimators designed using $\sdet$ can have small variance when the algebra dimension is sufficiently large. However, when the dimension is constant (the only regime where efficient algorithms for $\sdet$ are known) we do not get a polynomial-time approximation scheme for the permanent. In fact, their analysis suggests substantial obstacles to this program. This motivated us to study the symmetrized determinants in more detail: its algebraic properties, and its computability under algebras of non-constant dimension.

In general, since $\sdet$ is defined over an algebra, notion of invertibility is not well-defined anymore. Further, due to the non-commutativity of symmetrized determinants, properties of the standard determinants like multiplicativity (i.e, $\det AB = \det A \det B$) no longer hold. However, determinants are also known to follow many useful properties, 
one of them being the \emph{principal minor expansion}.
\begin{proposition}[Principal Minor Expansion Property \cite{horn2012matrix}]\label{prop:pme}
    Given a $n \times n$ matrix $A$ over a field $\F$, $\det (A+I_n) = \sum_{S \subseteq [n]} \det A_S $, where $I_n$ is the $n \times n$ identity matrix and $A_S$ is the principal submatrix of $A$ indexed by the elements in $S$. 
\end{proposition}
This property has found applications in varied settings. For instance, it serves as the \emph{normalization} constant behind $L$-ensemble determinantal point processes ($\mathsf{DPP}$s) which are used extensively in machine learning for recommender systems, clustering, and many others, see \cite{kulesza2012determinantal} for a survey. More concretely, we can interpret any $n \times n$ positive semi-definite matrix $L$ as a probability distribution over $n$ binary random variables using this identity as computing the normalization constant can now be done efficiently, since it is equivalent to computing the determinant of a single matrix instead of the exponential sum over all subsets. Not only that, in combinatorics, a direct application of this property for the Laplacian matrix of a graph counts the number of spanning rooted forests in the graph \cite{chebotarev2006matrix}. It's natural to question whether similar properties hold for the $\sdet$?

Besides, non-commutative algebraic computation has long been known to behave very differently from the commutative world: In \cite{nisan1991lower}, Nisan showed an exponential lower bound for computing the non-commutative determinant over the \emph{free algebra} via an Algebraic Branching Program ($\mathsf{ABP}$). His work was extended to a much larger class of non-commutative algebras by Chien and Sinclair in \cite{chien2011almost}. In \cite{arvind2010hardness}, Arvind and Srinivasan showed that small circuits for the non-commutative {Cayley} determinant would imply $\VP=\VNP$, hinting at the potential non-existence of such a small circuit. This was further strengthened when Bl\"aser \cite{blaser2015noncommutativity} showed that computing the {Cayley} determinant is $\sharpP$-hard if a specific quotient ring of the given algebra is non-commutative. In particular, this implies computing it is hard even over finite-dimensional matrix algebras, which is in sharp contrast to an efficient polytime algorithm for symmetrized determinants in these settings \cite{Barvinok2000NewPE}. Not only that, there has also been work on extending algebraic complexity classes like $\VP$ and $\VNP$ to non-commutative settings. Hrube\v{s} et al.\ \cite{hrubevs2010relationless} developed a general completeness theory for algebraic models without commutativity and/or associativity, showing that permanent-type families remain complete in relationless settings and that determinant-type families retain significant expressive power. More recently, Arvind et al.\ \cite{arvind2016noncommutative} studied these non-commutative algebraic classes and established complete problems in these contexts. These results naturally motivate the study of symmetrized determinants from a complexity perspective over larger algebras.
\vspace{-1em}
\subsection{Our results}
In this work, we study the algebraic behavior and computational properties of the symmetrized determinants. As our first result, we show that the \emph{principal minor expansion} property holds for the symmetrized determinants, see Theorem~\ref{thm:pma-expansion}. The proof relies on a careful combinatorial rearrangement in the expansion of the $\sdet$ formula in~(\ref{eq:sdet}) along with exploiting some properties exhibited by a pair of permutations. 
Thus, using Barvinok's result over fixed-dimension algebras it implies efficient computation of the analogue of \emph{normalization} constant of $\mathsf{DPP}$s and opens possibilities to exploit symmetrized determinants as probability distributions in machine learning and beyond.

Second, in Theorem~\ref{thm:sharpP-hardness} we show that computing the symmetrized determinant is $\sharpP$-Hard when the input matrix has entries in a suitably chosen non-commutative algebra whose dimension is polynomial in the size of the input. This is in contrast with the polynomial time algorithm for algebras of constant dimension. To prove the theorem, we interpret the matrix as an adjacency matrix of a graph with corresponding edge weights. We then define an algebra which ensures that all monomials corresponding to non-cycles in the graph go to $0$, while all cycles of the same length are grouped together. We can then look at the coefficient of size $n$ cycles to recover the number of Hamiltonian cycles. An implication of our result is the following: Moore and Russell \cite{moore2012approximating} suggested a possible approach towards getting an algebraic polynomial-time approximation scheme that relied on finding polynomial time algorithms for $\sdet$ when the algebra dimension is $\mathsf{poly}(n)$. While we only show hardness for a particular polynomial-dimension algebra, it unfortunately hints towards potential roadblocks in that approach.

Third, in Theorem~\ref{thm:vnpac-complete} we prove that the symmetrized determinant is $\VNP$-complete over a similar non-commutative algebra of polynomial dimension, matching the behaviour of the \emph{Cayley} determinant over such algebras \cite{arvind2010hardness}. Here, the $\VNP$ completeness is over the associative, non-commutative setting ($\vnpac$) as defined in \cite{hrubevs2010relationless}. For the proof, we use a similar algebra as above and additionally, extend the proof of $\VNP$-completeness of the Hamiltonian Cycle  polynomial to the associative, non-commutative setting, see Theorem~\ref{thm:hcac-vpnac-complete}. Finally, we reduce the $\sdet$ polynomial to it. This proof can also be extended to show that the symmetrized determinants is $\VNP$-complete as a family of commutative polynomials, when evaluated over a matrix algebra of polynomial dimension, see Theorem~\ref{thm:vnpc-complete}.

\emph{Author's note: After submission, the authors became aware that analogues of Theorem~\ref{thm:sharpP-hardness} and Theorem~\ref{thm:vnpac-complete} were previously established in \cite[Theorem 19, 20]{arvind2010hardness}; we present our independent proofs here for completeness.}

\section{Notation and Preliminaries}
We use $\alg$ to denote an arbitrary \emph{unital} (i.e., it has a multiplicative identity) algebra defined over a field of characteristic $0$. Let the multiplicative identity of $\alg$ be $1_{\alg}$. $\mat(\alg,p,q)$ denotes the set of all matrices of size $p \times q$ where each entry comes form $\alg$. Note that when $p=q=n$, $\mat(\alg,n,n)$ forms an algebra. $I_{\mat(\alg,n,n)}$ denotes the identity matrix in $\mat(\alg,n,n)$ where diagonal entry is $1_{\alg}$ and each off-diagonal entry is $0$. $A,B$ are used to denote arbitrary matrices from $\mat(\alg,n,n)$. $\sdet A$ is the function computing the symmetrized determinant of $A$: it takes a matrix $A \in \mat(\alg,n,n)$ as input and outputs an element in $\alg$. Unless otherwise stated, we use $\alpha, \sigma, \tau$ to denote permutations from $\Sym_n$, the symmetric group on $n$ elements. The sign of a permutation, denoted $\sgn(\sigma)$, counts the number of inversions in $\sigma$. That is, the sign is positive if the number of pairs $(i,j)$ such that $i < j$ but $\sigma(i) > \sigma(j)$ is even, otherwise odd.

Symmetrized determinants are similar to the ordinary determinants in a few ways, for example, $(\sdet A)^T = \sdet A^T$, Similarly, we can show the following:
\begin{proposition}
    Let $M \in \mat(\alg,p+q,p+q)$ such that $M = \begin{pmatrix}
            A_{p \times p} & 0_{p \times q} \\
            B_{q \times p} & D_{q\times q}
        \end{pmatrix}$. Then, $\sdet M = \sdet M'$ where $M' = \begin{pmatrix}
            A_{p \times p} & 0_{p \times q} \\
            0_{q \times p} & D_{q\times q}
        \end{pmatrix}$.   
\end{proposition}

\begin{proof}
    We know that any monomial in $\sdet M$ has exactly one elt from each row and column. Lets look at some monomial $m \in \sdet M$. If $m$ contains some term from $B$, then it can only index $p-1$ columns of $A$ and hence only $p-1$ rows of $A$, and hence one term would have to appear from a row in $0_{p \times q}$.
\end{proof}

However, they also have many differences as compared to the ordinary determinants which arise primarily due to non-commutativity. Most prominently, it can be easily verified that $\sdet AB \neq \sdet A \ \sdet B$, i.e, $\sdet$ is not \emph{multiplicative}. Hence, even simple decompositions like the \emph{Schur} complement cease to follow. Further, the standard method to compute determinants efficiently relies on row and column operations which inherently use the invertibility of field elements, and may not hold for elements of an arbitrary algebra. Similarly, other division-free combinatorial methods for computing the determinant like the \emph{Clow-sequence} method describe by Mahajan and Vinay \cite{mahajan-vinay} relies on the commutativity of the matrix entries. Despite these challenges, using clever identities involving \emph{mixed discriminants} \cite{barvinok1997computing}, Barvinok gave an efficient algorithm to compute the symmetrized determinants over algebras of a constant dimension.

\begin{theorem}(\cite[Theorem 3.6]{Barvinok2000NewPE})\label{thm:sdet-efficient-algo}
    Let $r$ be the dimension of the algebra. For a fixed $r$, given an $n \times n$ 
matrix $A \in \mat(\alg, n,n)$, we can compute $\sdet A \in \alg$ using $O(n^{r+3})$ arithmetic operations. 
\end{theorem}

\section{Principal Minor Expansion Property}
In this section, we show that the principal minor expansion property also holds for the symmetrized determinants. But before we prove that, we will first prove a property about a pair of permutations that will be useful later.

\begin{lemma}\label{lem:mixed-inversion}
    Let $\sigma, \tau \in \Sym_n$. Let $S$ be a fixed subset of $[n]$. Suppose, $\sigma(i) = \tau(i)$ for all elements in $C \coloneq [n]\setminus S$. Let $d_{\sigma}$ count the number of inversions in $\sigma,$ such that $i \in S, j \in C$ and we have: either $i <j$ $\&$ $\sigma(i)>\sigma(j)$, or $i >j$ $\&$ $\sigma(i) < \sigma(j)$. Then, $d_{\sigma} = d_{\tau} \pmod 2$.  
\end{lemma}

\begin{proof}
     For each $c\in C$, let $B_\sigma(c)$ denote the number of ``mixed'' inversions involving $c$ under $\sigma$, i.e. $ B_\sigma(c) = |\{s\in S : s<c,\ \sigma(s)>\sigma(c)\} |+|\{s\in S : s>c,\ \sigma(c)>\sigma(s)\}|$. Then, $d_{\sigma} = \sum_{c\in C} B_\sigma(c)$. Defining $B_{\tau}(c)$ analogously, we have $d_{\tau} = \sum_{c\in C} B_\tau(c)$. 

     Since $\sigma(i) = \tau(i)$ for all $i \in C$, let $\sigma(S)=\tau(S) = G \subseteq [n]$. Fix $c\in C$. Since $\sigma(S)=G$ and $\sigma$ is a permutation, we have $\sigma(c)\notin G$. Let $L_c$ denote the number of elements to the left of $c$ in $S$, and $L'_c$ denote the number of elements to the left of $\sigma(c)$ in $G$. That is, $L_c := |\{s\in S : s<c\}|$ and $L'_c := |\{g\in G : g<\sigma(c)\}|$. Notice that $L_c$ and $L'_c$ are independent of the permutation as $\sigma(c) = \tau(c)$.
     
     Let $x_{\sigma}(c)$ denote the number of non-inversion elements in $S$ to the left of $c$. Formally, $x_\sigma(c):=|\{s\in S : s<c,\ \sigma(s)<\sigma(c)\}|$. Then among the $L_c$ elements of $S$ lying to the left of $c$, exactly $x_\sigma(c)$ have image $<\sigma(c)$, so
\[
|\{s\in S : s<c,\ \sigma(s)>v\}|=L_c-x_\sigma(c).
\]
On the other hand, since $\sigma(S)=G$, exactly $L'_c$ elements of $S$ are mapped by $\sigma$ to values $<\sigma(c)$. Of these, $x_\sigma(c)$ lie to the left of $c$, so
\[
|\{s\in S : s>c,\ \sigma(s)<\sigma(c)\}|=L'_c-x_\sigma(c).
\]
Therefore
\[
B_\sigma(c)
=
(L_c-x_\sigma(c))+(L'_c-x_\sigma(c))
=
L_c+L'_c-2x_\sigma(c),
\]
and hence $B_\sigma(c)\equiv L_c+L'_c \pmod 2$. Exactly the same argument applied to $\tau$ gives $B_\tau(c)\equiv L_c+L'_c \pmod 2$. Thus $B_\sigma(c)\equiv B_\tau(c)\pmod 2$ for every $c \in C$. Hence, $d_{\sigma} = \sum_{c\in C} B_\sigma(c) \equiv
\sum_{c\in C} B_\tau(c) \pmod 2 = d_{\tau} \pmod 2$.
\end{proof}
We now move on towards proving that the symmetrized determinants also follow the principal minor expansion property followed by the determinants (see Proposition~\ref{prop:pme}).

\begin{theorem}[Principal minor expansion property for $\sdet$]\label{thm:pma-expansion}
Let $A \in \mat(\alg, n,n)$ and $I$ be the identity matrix in $\mat(\alg,n,n)$. Then,
\[
\sdet(A+I) = \sum_{S \subseteq [n]} \sdet(A_S)
\]
where $A_S$ is the principal submatrix of $A$ formed by taking only those rows and columns indexed by elements of some subset $S \subseteq [n]$.
\end{theorem}

\begin{proof}
    Let $A = (a_{ij})$ and $I = (\Delta_{ij})$, where $\Delta_{ij} = 0_{\alg}$ for $i \neq j$, and $\Delta_{ii} = 1_{\alg}$. Using (\ref{eq:sdet}), we have
    \begin{equation}\label{eq:sdet-pf-1}
        \sdet (A + I) = {1 \over n!} \sum_{(\sigma, \tau) \in \Sym_n \times \Sym_n} 
    (\sgn \sigma)(\sgn \tau) \prod_{j=1}^n  (a_{\sigma(j) \tau(j)} + \Delta_{\sigma(j) \tau(j)})
    \end{equation}
    Now, $0_{\alg} \text{ and } 1_{\alg}$ both commute with all elements in $\alg$. Hence, we can rewrite
    \[
    \prod_{j=1}^n  (a_{\sigma(j) \tau(j)} + \Delta_{\sigma(j) \tau(j)}) = \sum_{S \subseteq [n]} \Bigg( \prod_{i \in S} a_{\sigma(i) \tau(i)} \prod_{i \notin S} \Delta_{\sigma(i) \tau(i)}\Bigg)
    \]
    Replacing in (\ref{eq:sdet-pf-1}), this gives
    \begin{equation}\label{eq:sdet-pf-2}
    \begin{split}
        \sdet (A + I) &= {1 \over n!} \sum_{(\sigma, \tau) \in \Sym_n \times \Sym_n} (\sgn \sigma)(\sgn \tau) \sum_{S \subseteq [n]} \Bigg( \prod_{i \in S} a_{\sigma(i) \tau(i)} \prod_{i \notin S} \Delta_{\sigma(i) \tau(i)}\Bigg) \\
        &= \sum_{S \subseteq [n]} {1 \over n!} \sum_{(\sigma, \tau) \in \Sym_n \times \Sym_n} (\sgn \sigma)(\sgn \tau) \Bigg( \prod_{i \in S} a_{\sigma(i) \tau(i)} \prod_{i \notin S} \Delta_{\sigma(i) \tau(i)}\Bigg)
    \end{split}
    \end{equation}
    
    Now, we know that $\Delta_{i,j} = 0_{\alg}$ whenever $i \neq j$. Hence, in the above summand, only those terms with $\sigma(i) = \tau(i)$ will survive. Essentially, this means that $\tau$ is a permutation that matches $\sigma$ on all elements not in $S$. Thus, we have
    \[
    \sum_{(\sigma, \tau) \in \Sym_n \times \Sym_n} (\sgn \sigma)(\sgn \tau) \Bigg( \prod_{i \in S} a_{\sigma(i) \tau(i)} \prod_{i \notin S} \Delta_{\sigma(i) \tau(i)}\Bigg) = \sum_{\substack{(\sigma, \tau) \in \Sym_n \times \Sym_n\\\sigma(i) = \tau(i) \: \text{ for all } i \notin S}} (\sgn \sigma)(\sgn \tau) \Bigg( \prod_{i \in S} a_{\sigma(i) \tau(i)} \Bigg)
    \]
    Replacing this in (\ref{eq:sdet-pf-2}), we get
    \begin{equation}\label{eq:sdet-pf-3}
    \begin{split}
        \sdet (A + I) 
        &= \sum_{S \subseteq [n]} {1 \over n!} \sum_{\substack{(\sigma, \tau) \in \Sym_n \times \Sym_n\\\sigma(i) = \tau(i) \: \text{ for all } i \notin S}} (\sgn \sigma)(\sgn \tau) \Bigg( \prod_{i \in S} a_{\sigma(i) \tau(i)} \Bigg)
    \end{split}
    \end{equation}
        
    Let's fix $S \subseteq [n]$ such that $|S| = K$. 
    $\sgn \sigma$ is the number of inversions in $\sigma$. Those inversions can be decomposed into three parts: 
    \begin{enumerate}
        \item Inversions between elements in $S$, i.e, $i<j \in S$ but $\sigma(i) > \sigma(j)$. Let $I_{\sigma}^1$ count such inversions.
        \item Inversions between elements in $C\coloneqq [n] \setminus S$, i.e, $i < j \in C$ but $\sigma(i)>\sigma(j)$. Let $I_{\sigma}^2$ count such inversions.
        \item ``Mixed'' inversions between $S$ and $C$. Let $i \in S, j \in C$, then we have either $i < j, \sigma(i) > \sigma(j)$ or $i>j, \sigma(i)<\sigma(j)$. Let $I_{\sigma}^3$ count such inversions.
    \end{enumerate}
    Since $\sigma(i) = \tau(i)$ for all $i \in C$, it is easy to see that $I_{\sigma}^2 = I_{\tau}^2$. Further, by Lemma~\ref{lem:mixed-inversion}, we know that $I_{\sigma}^3 = I_{\tau}^3 \pmod 2$. Hence, \[
    \sgn \sigma \sgn \tau = (-1)^{I_{\sigma}^1+I_{\sigma}^2+I_{\sigma}^3} \cdot  (-1)^{I_{\tau}^1 + I_{\tau}^2 + I_{\tau}^3} = (-1)^{I_{\sigma}^1+I_{\tau}^1}
    \]
    Essentially, the only contribution in the sign comes from permutations within $S$. Hence, given a $S$ and a pair $(\sigma, \tau)$, there are $(n-k)!$ pairs $(\hat \sigma,\hat \tau )$ that mimic $(\sigma,\tau)$ on $S$ and are free to do anything on $C$. Since, sign only depends on behaviour on $S$, we have $\sgn \hat \sigma \sgn \hat \tau = \sgn \sigma \sgn \tau$. This gives
    \begin{equation}\label{eq:sdet-pf-4}
    \begin{split}
        \sdet (A + I) 
        &= \sum_{S \subseteq [n]} {1 \over n!}(n-k)! \sum_{\substack{(\sigma, \tau) \in \Sym_n \times \Sym_n\\ \sigma(i) = \tau(i) = i \: \text{ for all } i \notin S}} (-1)^{I_{\sigma}^1+I_{\tau}^1} \Bigg( \prod_{i \in S} a_{\sigma(i) \tau(i)} \Bigg)
    \end{split}
    \end{equation}
    Now, $\sigma,\tau$ are permutations in $\Sym_n$ that are fixed on $n-k$ elements not in $S$. But there are $n \choose k$ subsets $G$ of $[n]$ to which $\sigma, \tau$ can move the elements. Note that $\sigma(S) = \tau(S)$ since they match on all elements not in $S$. In particular, if $G$ is a subset of size $k$, then there are $k!^2$ pairs of $(\sigma,\tau) \in \Sym_n^2$ which are fixed on elements not in $S$ and have $\sigma(S) = \tau(S) = G$. Further, they can now be considered as permutations in $\Sym_k$ by simply using a map $\phi: G \rightarrow S$ such that $\phi(i) < \phi(j)$ for all $i \neq j \in G$. Since, the sign only depends on the inversions in $S$, these permutations considered as $k$-permutations will have the same sign. Finally, in (\ref{eq:sdet-pf-3}), the monomial only has terms $\prod_{i \in S} a_{\sigma(i)\tau(i)}$. This means that we are only taking a product over the elements of $A_{G}$ where $A_{G}$ is the submatrix of $A$ formed by taking only those rows and columns indexed by elements of $G$. Hence, we can rewrite \[
    \sum_{\substack{(\sigma, \tau) \in \Sym_n \times \Sym_n\\ \sigma(i) = \tau(i) = i \: \text{ for all } i \notin S}} (-1)^{I_{\sigma}^1+I_{\tau}^1} \Bigg( \prod_{i \in S} a_{\sigma(i) \tau(i)} \Bigg) = k!\sum_{\substack{G \subseteq [n]\\ |G|=k}} \frac{1}{k!} \sum_{\substack{(\sigma, \tau) \in Sym_k \times Sym_k\\ \sigma(S) = \tau(S) = G}} (-1)^{I_{\sigma}^1+I_{\tau}^1} \Bigg( \prod_{i \in S} a_{\sigma(i) \tau(i)} \Bigg)
    \] 
    \[
    = k! \sum_{\substack{G \subseteq [n]\\ |G|=k}} \sdet A_G
    \]
    Putting this in (\ref{eq:sdet-pf-4}) and using $k = |S|$, we get 
    \begin{equation}\label{eq:sdet-pf-5}
        \sdet (A+I) = \sum_{S \subseteq [n]} {1 \over n!}(n-|S|)! |S|! \sum_{\substack{G \subseteq [n]\\ |G|=|S|}} \sdet A_G = \sum_{S \subseteq [n]} \sum_{\substack{G \subseteq [n]\\ |G|=|S|}} \frac{1}{{n \choose |S|}} \sdet A_G
    \end{equation}
    But this simply means that given any $|G|$ of size $|S|$, we get $n \choose |S|$ summands including $\sdet A_G$ because every $S \subseteq [n]$ of that size contributes one such term. Hence,\[
    \sdet (A+I) = \sum_{G \subseteq [n]} \sdet A_G
    \]
    \vspace{-2em}
\end{proof}
Hence, using this result, over finite-dimensional algebras we can compute the analogue of the normalization constant for $\mathsf{DPP}$s in polynomial time.

\section{$\sharpP$-Hardness of symmetrized determinants}
In this section, we prove the $\sharpP$-hardness of the symmetrized determinants over a fixed algebra with dimension polynomial in the size of the input. This is contrast with the known polytime algorithm for algebras of a constant dimension shown in \cite{Barvinok2000NewPE}, and aligns more closely with the $\sharpP$-hardness of Cayley determinant \cite{blaser2015noncommutativity}.

One approach to show this hardness is to use the exterior algebra over $2n$ generators (where the input matrix $A \in \mat(\alg,n,n)$). This algebra allows us to reduce computing the symmetrized determinants to computing the number of perfect matchings over a graph which is a known $\sharpP$-complete problem \cite{arora2009computational}. However, the issue is that the algebra blows up despite having only $2n$ generators, and in fact, $dim(\alg) = 2^{2n}$. Thus, we only get $\sharpP$-Hardness under a sparse encoding. In fact, number of perfect matchings of a bipartite graph can be counted in $\tilde O(2^{n/2})$ time \cite{bjorklund2012counting} and is trivially in $\mathsf{P}$ under the dense encoding. Instead, we use the $\sharpP$-completeness of counting the number of Hamiltonian Cycles \cite{arora2009computational} to get our desired reduction. 

\subsection{Defining the algebra}
To this end, we define an associative algebra $\alg$ over a field $\mathbb{F}$ such that the dimension of the algebra is $ \leq 2(n+1)n^2$, with the basis being elements of the form $e=t^{\alpha} s^{\beta} u_{i,j}$ where $\alpha \in \{0,1\}, \beta \in \{1,...,n\}$ and $1\leq i,j \leq n$. The elements of the algebra follow these rules:
\begin{enumerate}
    \item \[
            u_{i,j}\cdot u_{k,l} = \begin{cases}
                    0 & \text{if } j \neq k, \\
                    u_{i,l} & \text{otherwise}.
            \end{cases}
        \]
    \item $t^2 = t$
    \item $s\cdot u_{i,j} = u_{i,j} \cdot s$ and $t\cdot u_{i,j} = u_{i,j} \cdot t$
    \item $s\cdot t = 0$ and $t \cdot s \neq 0$.
\end{enumerate}
We define the set of generators as $\gen =\{t, \ s\cdot u_{1,2},s\cdot u_{1,3},...,s\cdot u_{1,n}, \ s\cdot u_{2,1}, s\cdot u_{2,3}, \cdots, s\cdot u_{2,n}, \  \cdots,s\cdot u_{n,1}\cdots, s\cdot u_{n,n-1}\}$. Notice that $(s \cdot u_{i,j}) \cdot t = u_{i,j} \cdot (s \cdot t) = 0$ and $t \cdot (s\cdot u_{i,j}) = tsu_{i,j}$. Further, $su_{i,j}\cdot su_{j,k} = su_{i,l}$. Hence, it is easy to see that $\langle \gen \rangle \subseteq \alg$. In fact, this can be realized using a matrix algebra.

\subsubsection{Instantiation of the algebra} Let the base algebra be $\alg = \mat(\Q,n+1,n+1) \otimes \mat(\Q,n,n)$. Then, to satisfy the above properties, we can set the elements as following: \begin{itemize}
    \item $u_{i,j} = I_{n+1} \otimes e_{i,j}$
    \item $s = \big( \sum_{i=1}^n e_{i,i+1} \big) \otimes I_n$
    \item $t = e_{1,1} \otimes I_n$
\end{itemize}
Here $e_{i,j}$ is has $1$ only in the position $i,j$ and zeros everywhere else. It is easy to verify this instantiation satisfies all the properties we desired in our algebra.

\subsection{Towards reduction}
We now prove our main theorem regarding the hardness of the symmetrized determinant.
\begin{theorem}[$\sharpP$-hardness for $\sdet$]\label{thm:sharpP-hardness}
There exists an algebra $\alg$ of dimension $\text{poly}(n)$ such that computing $\sdet A$ where $A \in \mat(\alg,n,n)$ is $\sharpP$-Hard.
\end{theorem}
\begin{proof}
    We show the reduction from counting number of Hamiltonian Cycles in an undirected graph using the algebra $\alg$ defined in the previous section. Given an undirected graph $G=(V,E)$ with $|V|=n$. Let an edge $\eps_{ij}=(v_i,v_j) \in E$. Given such a graph $G$, we construct the $n \times n$ adjacency matrix $M_G$ as follows: \begin{itemize}
    \item Every entry of $M_G$ is an element from $\gen \subseteq \alg$. Further $(M_G)_{i,j} = (M_G)_{j,i} = 0$ if there is no edge between $v_i,v_j$. 
    \item If $\eps_{ij}$ exists, then, $(M_G)_{i,j} = s\cdot u_{i,j}$ and $(M_G)_{j,i} = s\cdot u_{j,i}$.    
    \item Set all diagonal entries of $M_G$ as $t$.
\end{itemize}

Recall,
\[
\sdet M_G={1 \over n!} \sum_{(\sigma, \tau) \in \Sym_n \times \Sym_n} 
(\sgn \sigma)(\sgn \tau) \prod_{i=1}^n (M_G)_{\sigma(i), \tau(i)}
\]

Now, consider any term $h = \prod_{i=1}^n (M_G)_{\sigma(i), \tau(i)}$ in $\sdet M_G$. We know that $(M_G)_{\sigma(i), \tau(i)} \in \{0,t,su_{a,b}\}$ by construction. Thus it can only be non-zero if $(M_G)_{\sigma(i), \tau(i)} \neq 0 \text{ for all } i \in [n]$. Further, in $\alg$, $s \cdot t = 0$, and $s$ and $t$ commute with $u_{i,j}$. Hence, if $t$ appears after some $su_{i,j}$ then $h$ becomes $0$. Thus for $h\neq 0$, in $h$ all occurrences of $t$ must be before occurrences of $su_{i,j}$. Wlog, suppose there are $k$ occurrences of $su_{a,b}$ for some $a,b$. Then we must have $n-k$ occurrences of $t$ occurring first, i.e, $\text{ for all } i \in [n-k]$, we have $\sigma(i) = \tau(i)$ as $t$ only appears on the diagonals of $M_G$. Hence, the permutation $\alpha = \tau \circ \sigma^{-1}$ keeps $j$ fixed for  $n-k$ many $j$ $\in [n]$. \[h = t^{n-k} \cdot s u_{\sigma(n-k+1),\tau(n-k+1)} \cdots s u_{\sigma(n),\tau(n)} = t^{n-k}s^k u_{\sigma(n-k+1),\tau(n-k+1)}\cdots u_{\sigma(n),\tau(n)}\] using commutativity of $s$ with $u_{i,j}$. 

Let $\sigma(n-k+j)  = l_j \text{ for all } j \in [k]$, and $\tau(n-k+j)  = r_j \text{ for all } j \in [k]$. Clearly $l_i \neq l_j, r_i \neq r_j$ for $i\neq j$. Thus, we have $h = t^{n-k}s^k u_{l_1,r_1} \cdots u_{l_k,r_k}$ 
Now, by the rules of multiplication of $u_{a,b}$, for $h \neq 0$ it must hold that $r_1 = l_2, r_2 = l_3 ,..., r_{k-1} = l_k$. In other words, $\tau(n-k+j)  = l_{j+1} \text{ for all } j \in [k-1]$. Thus, $n-k+j = \sigma^{-1}(l_j)$ and hence $l_{j+1} = \tau(n-k+j) = \tau \circ \sigma^{-1}(l_j) \text{ for all } j \in [k-1]$. Thus, $\alpha(l_j) = l_{j+1}$ for $j \in [k-1]$. Now, $\alpha(l_k) = \tau \circ \sigma^{-1}(l_k) = \tau (n) = r_k$. But since we know that $\alpha$ keeps $n-k$ other entries fixed, $r_k = l_1$. Hence, $\alpha \in \Sym_n$ is a $k$-cycle given by $\alpha = (l_1\; l_2 \; ... \; l_k)$. Further, we now have \[
h = t^{n-k}s^k u_{l_1,l_1} = \begin{cases}
                    s^n u_{l_1,l_1} & \text{if } k=n, \\
                    ts^k u_{l_1,l_1} & \text{otherwise}.
            \end{cases}
\]
Hence, any non-zero term of $\sdet M_G$ corresponds exactly to a $k$-cycle $\alpha \in \Sym_n$ for some $k$, such that $\alpha = \tau \circ \sigma^{-1}$ and $\sgn \alpha = \sgn \sigma \sgn \tau$. Further $h$ is exactly a basis element. Given a fixed $\alpha$, there are $n!$ pairs $(\sigma, \tau)$ such that $\alpha = \tau \circ \sigma^{-1}$. However, for $h$ to be non-zero, we also need $\sigma(i) = \tau(i) \text{ for all } i \in [n-k]$. Thus, if $\alpha = (l_1 \cdots l_k)$, then $\sigma(i) = \tau(i) = j$ for $j \in [n] \setminus \{l_1,...,l_k\}$. Thus, we can choose the different $j$ in $(n-k)!$ ways for a fixed $\alpha$. Hence, out of the $n!$ pairs, only $(n-k)!$ of them give a non-zero term of the $\sdet$. For the special case when $k=0$, we are simply choosing all diagonal elements and there are $n!$ pairs $(\sigma, \sigma)$ doing that. Further, $k=1$ is not possible as using Equation~(\ref{eq:sdet}), choosing $n-1$ diagonal entries forces us to choose all $n$ diagonal entries. Finally, for a fixed $\alpha = (l_1 \cdots l_k)$, as we saw before, $h \neq 0$ if $\tau(n-k+j) = \sigma(n-k+(j+1))$ and $\tau(n) = \sigma(n-k+1)$. Thus, given $\alpha$, for the above condition to hold, we can choose $\sigma(n-k+1)$ in $k$ ways as a cycle can start from any position. Further, since the graph is undirected, both the cycle and its inverse correspond to the same non-zero terms, and hence having chosen $\sigma(n-k+1)$, we can choose $\sigma(n-k+2)$ in $2$ ways. After this, the cycle gets fixed and there is a single choice for $\sigma(n-k+j)$ for $j \in \{3,...,k\}$. It is easy to see that out of these $2k$ choices, we have exactly two terms $t^{n-k}s^k u_{l_i,l_i}$ for each $i \in [k]$. Hence, \[
\sdet M_G={1 \over n!} \sum_{(\sigma, \tau) \in \Sym_n \times \Sym_n} 
(\sgn \sigma)(\sgn \tau) \prod_{i=1}^n (M_G)_{\sigma(i), \tau(i)}
\]
\[
= 
{1 \over n!} \sum_{\substack{(\sigma, \tau) \in \Sym_n \times \Sym_n \\ \alpha = \tau \circ \sigma^{-1} \text{ is a $k$-cycle}}} 
\sgn \alpha \prod_{i=1}^n (M_G)_{\sigma(i), \tau(i)}
\]
\begin{align}\label{eq:sdet-Alg-reduction}
    = t + {1 \over n!} \sum_{\substack{ l_1 - ... - l_k \text{ is a cycle in $G$} \\ k \in \{2,...,n\}}} (-1)^{k+1} 2(n-k)! \; t^{n-k}s^k (u_{l_1,l_1} + ... + u_{l_k,l_k})  
\end{align}

In particular, the coefficient of $s^n$ is exactly \[{2 \over n!}(-1)^{n+1} (u_{l_1,l_1} + ... + u_{l_n,l_n}) \sum_{l_1 - ... - l_n \text{ is a cycle in $G$}}1\]
Since each $n$-cycle must contain every vertex and every vertex must be present in each $n$-cycle, this means for all the $n$-cycles, $u_{1,1}$ will appear in the coefficient of $s^n$. Thus, the coefficient $C$ of $s^nu_{1,1}$ is \[
\frac{2R}{n!}(-1)^{n+1}
\]
where $R =$ number of $n$-cycles (or Hamiltonian cycles) in $G$. Since, $s^nu_{1,1}$ is a basis elements and we have $\leq 2(n+1)n^2 = B$ basis elements, we can maintain a tuple of size $B$ in $poly(n)$ time to keep track of this coefficients, and use it to calculate the number of Hamiltonian cycles $R = \frac{(-1)^{n+1}Cn!}{2}$. Since, calculating number of Hamiltonian cycles is a well-known $\sharpP$-complete problem \cite{arora2009computational}, this shows that computing $\sdet$ over an algebra of size $O(n^4)$ solves this problem.
\end{proof}
\begin{remark}
    Our reduction counts number of cycles of length $k$ for \emph{each} $k$ by scaling the corresponding basis coefficient with a suitable constant, which might be of independent interest.
    A simpler reduction is to substitute $t=0$ and $s=1$ in the algebra definition. This will give us a smaller algebra which still counts the number of Hamiltonian Cycle.
\end{remark}

\section{$\VNP$-completeness of the symmetrized determinant polynomial}
\subsection{In the non-commutative setting}
In this section, we show that there exists an algebra of dimension polynomial in the input size such that the symmetrized determinants are a complete class of polynomials for $\VNP$ in the associative, non-commutative setting. Originally, algebraic classes like $\VP$ and $\VNP$ were defined only in the associative, commutative setting \cite{valiant1979completeness}. Naturally, the definitions are very amenable to extension to non-commutative settings. In fact, there has been a lot of work on this, see \cite{nisan1991lower,hrubevs2010relationless,arvind2010hardness}. We use the standard definitions for the classes in this setting, following Hrube\v{s} et al \cite{hrubevs2010relationless}. In particular, we focus on $\vpac, \text{ and } \vnpac$, the class of polynomials where the multiplication between the variables is associative, but not commutative. We prove that the $\sdet$ polynomial is complete for the class $\vnpac$ under $c$-reductions which are an analogue of polynomial \emph{Turing} reductions in the algebraic setting, see \cite[Definition 5.14]{burgisser2000completeness}, and is a commonly used notion to show reductions in the algebraic world. Hrube\v{s} et al. \cite[Fact 1]{hrubevs2010relationless} showed that a polynomial family $\{f_n\}$ is in $\vnpac$ if given a monomial, its coefficient can be evaluated in polynomial time. Hence, we can conclude the following.
\begin{proposition}\label{prop:sdet-in-vnpac}
    $\sdet$ is in $\vnpac$.
\end{proposition}
Hence, to show that $\sdet$ is $\vnpac$-complete, we need to show that it is hard for the class. To achieve this, we show a reduction from the Hamiltonian Cycle polynomial ($\hcac$) in this setting. While the Hamiltonian Cycle polynomial is known to be $\VNP$-complete in the usual setting \cite[Theorem 3]{valiant1979completeness}, we prove this also hold in the associative, non-commutative setting. We start by defining the polynomial in this setting. \begin{definition}\label{def:hcac-poly}[Hamiltonian Cycle Polynomial in associative, non-commutative setting]
    Let $G=(V,E)$ be a directed graph with $V=\{v_1,\ldots,v_n\}$ having a fixed ordering on the vertices. We denote the the adjacency matrix of $G$ as $A_G$. The $(i,j)$-th entry of $A_G$ is denoted as $a_{i,j}$. Then, we define $A_G$ as follows: $a_{i,j} = 0$ if there is no edge from $v_i$ to $v_j$, and $a_{i,j} = x_{i,j}$ otherwise. Here, the $x_{i,j}$ are non-zero elements of the \emph{free} non-commutative algebra. Let $C_n$ denote the set of permutations in $\Sym_n$ with exactly one cycle of length $n$. Then, we define the Hamiltonian cycle polynomial as follows. \[
    \hcac = \sum_{\sigma \in C_n} x_{1,\sigma(1)}x_{\sigma(1), \sigma(2)} \cdots x_{\sigma(n-1), 1}
    \]
\end{definition}
Using this definition, we can establish the following.
 
\begin{theorem}\label{thm:hcac-vpnac-complete}
    $\hcac$ is $\vnpac$-complete under $p$-projections.
\end{theorem}
The proof of this theorem uses analogously defined ``rosette'' and ``glue'' gadgets as used in the proof of $\VNP$-completeness of permanent in \cite{burgisser2000completeness}. In the usual (associative, commutative) setting such gadgets are known for the Hamiltonian Cycle polynomial in folklore. However, we couldn't find any work where this is formally proven. For completeness, we prove these constructions work and extend the notion to the associative, non-commutative setting.
We defer the proof of this theorem to Appendix~\ref{app:hcac-completeness}. Using this theorem, we can now establish our main result: the completeness of symmetrized determinants in this setting.
\begin{theorem}\label{thm:vnpac-complete}
    $\sdet$ is $\vnpac$-complete under $c$-reductions.
\end{theorem}
\begin{proof}
    Our proof uses an algebra very similar to the one used in proof of Theorem~\ref{thm:sharpP-hardness}. In this case, we set $t=0$ and $s=1$. Explicitly, the algebra is simply $\alg = \mat(\Q,n,n)$ which is generated by $\{e_{i,j}\}_{i,j \in [n]}$ where $e_{i,j}$ is a $n \times n$ matrix $1$ in the $(i,j)$-position and zeros everywhere else. Given a directed graph $G = (V,E)$ where $V = \{v_1,\ldots,v_n\}$, we now proceed to construct a $n \times n$ matrix $M$ as follows: \begin{enumerate}
    \item Set all diagonal entries to $0 \otimes 0 \otimes 0$.
    \item For $j \in [n]$, set $M_{1,j} = f \otimes e_{1,j} \otimes x_{1,j} $ where $x_{1,j}$ is non-zero and element of the \emph{free} non-commutative algebra if the graph $G$ had an edge from $1$ to $j$. Otherwise set $M_{1,j} = 0 \otimes 0 \otimes 0$.
    \item For $i \neq 1,$ $j \in [n]$, set $M_{i,j} = s \otimes e_{i,j} \otimes x_{i,j} $ where $x_{i,j}$ is non-zero and element of the free non-commutative algebra if the graph $G$ had an edge from $i$ to $j$. Otherwise set $M_{i,j} = 0 \otimes 0 \otimes 0$.
    \item Set all other entries as $0 \otimes 0 \otimes 0$.
\end{enumerate}
 Here $f = \begin{pmatrix}
    0 & 1 \\
    0 & 0
\end{pmatrix}$ and $s = \begin{pmatrix}
    0 & 0 \\
    0 & 1
\end{pmatrix}$. Intuitively, $f,s$ ensure that any Hamiltonian Cycle in $\sdet M$ always begins from node $1$ as $f \cdot s = f$ but $s \cdot f = 0$. Recall, \[
\sdet M = \frac{1}{n!} \sum_{\sigma, \tau \in \Sym_n \times \Sym_n} \prod_{i=1}^n m_{\sigma(i),\tau(i)}
\]
We will now show that many of these terms are actually the zero tensor. First, for any permutation pair $(\sigma, \tau)$ if $\sigma(i) = \tau(i)$ for any $i \in[n]$, this implies we are looking at the diagonal entries. Since diagonal entries are $0$ by construction, all such permutation pairs don't contribute any term to $\sdet M$. Further, since $s \cdot f = 0$, any non-zero term must have $f$ in the beginning followed by $s$. Hence, every non-zero term must have the first term of the form $m_{1,j}$ for some $j \in [n]$ such that there was a directed edge from $1$ to $j$ in $G$. Thus, $\sigma(1) = 1$ for all non-zero terms. Now we have $e_{i,j}$ in the middle tensor. Since $e_{i,j}\cdot e_{k,l}$ is non-zero only when $j=k$. This ensures that all non-zero terms must follow the permutation order: $\tau(i) = \sigma(i+1)$ for $i \in [n-1]$. Since $\sigma(1)=1$ in any non-zero terms, this ensure $\tau(n)=1$. Hence, the middle tensor in any non-zero term in $\sdet M$ will be $e_{1,1}$. Also, this ensures that all non-zero terms actually correspond to directed Hamiltonian Cycles in the graph. Analogous to the proof of Theorem~\ref{thm:sharpP-hardness}, we can see that there exists a $n$ cycle $\alpha$ in $\Sym_n$ such that $\sgn \sigma \sgn \tau = \sgn \alpha = (-1)^n$. Thus, any non-zero term is of the form $\frac{1}{n!} (-1)^n \big(f \otimes e_{1,1}  \otimes (x_{1,j_1}x_{j_1,j_2}\cdots x_{j_{n-1},1})\big)$. Here the rightmost part corresponds to the directed Hamiltonian Cycle $h$ given by $1 \rightarrow j_1 \rightarrow \ldots \rightarrow j_{n-1} \rightarrow 1$. 

Similarly, for any Hamiltonian cycle $h = 1 \rightarrow j_1 \rightarrow \ldots \rightarrow j_{n-1} \rightarrow 1$, it is easy to verify we get the non-zero term $\frac{1}{n!} (-1)^n \big(f \otimes e_{1,1}  \otimes (x_{1,j_1}x_{j_1,j_2}\cdots x_{j_{n-1},1})\big)$. Hence, there is a one to one correspondence between non-zero terms and Hamiltonian cycles. Thus, using Definition~\ref{def:hcac-poly}, we get \[
 \sdet M = \frac{1}{n!}(-1)^n (f \otimes e_{1,1} \otimes \hcac)
\]
Simply applying the operator $(-1)^n n! \big( \tr \otimes \tr \otimes 1 \big)$ to $\sdet M$ gives us $\hcac$ (here $\tr$ denotes the trace operator). Hence, $\hcac \leq_{c} \sdet$ under $c$-reductions. Combining with Proposition~\ref{prop:sdet-in-vnpac} and Theorem~\ref{thm:hcac-vpnac-complete}, we get that $\sdet$ is $\vnpac$-complete. 
\end{proof}

\subsection{In the commutative setting over the matrix algebra}
In this section, we will restrict our attention to the symmetrized determinant polynomial defined over matrix algebras, i.e., $\alg = \mat(\F,m,m)$ for some field $\F$ and $m \in \N$. We can now extend the definition of $\sdet$ polynomial to the commutative setting by defining a family of polynomials which together comprise it.
\begin{definition}\label{def:sdet-family-comm}
    Let $\alg = \mat(\F,m,m)$ and $M \in \mat(\alg,n,n)$, such that $M_{i,j} = m^{i,j} \in \alg$ where $(k,l)$-th entry of $m^{i,j}$ is denoted as $m^{i,j}_{k,l}$.
   
    Now, let $T = \sdet M \in \alg$. 
    
    We then define $\sdet_{m,n} M$ as the family of polynomials $\{T_{i,j}\}_{1\leq i,j \leq m}$. The $\{T_{i,j}\}$ are commutative polynomials over the variables $m^{i,j}_{k,l}$ of degree $n$.
\end{definition}
\begin{example}
    Let $\alg = \mat(\Q,2,2)$. Let $M \in \mat(\alg,2,2)$ such that \[
    \small M = \begin{pmatrix}
        \begin{pmatrix}
            a_{1,1} & a_{1,2} \\
            a_{2,1} & a_{2,2}
        \end{pmatrix} & \begin{pmatrix}
            b_{1,1} & b_{1,2} \\
            b_{2,1} & b_{2,2}
        \end{pmatrix} \\
        \begin{pmatrix}
            c_{1,1} & c_{1,2} \\
            c_{2,1} & c_{2,2}
        \end{pmatrix} & \begin{pmatrix}
            d_{1,1} & d_{1,2} \\
            d_{2,1} & d_{2,2}
        \end{pmatrix}
    \end{pmatrix}
    \]
    Then $\sdet M = \frac{1}{2}\begin{pmatrix}
        t_{1,1} &  t_{1,2} \\ 
        t_{2,1} &  t_{2,2}
    \end{pmatrix}$ where \begin{align*}
        t_{1,1} &= 2a_{1,1}d_{1,1} + a_{1,2}d_{2,1} + a_{2,1}d_{1,2} - 2b_{1,1}c_{1,1} - b_{1,2}c_{2,1} - b_{2,1}c_{1,2} \\
        t_{1,2} &= a_{1,1}d_{1,2} + a_{1,2}d_{2,2} + a_{1,2}d_{1,1} + a_{2,2}d_{1,2} - b_{1,1}c_{1,2} - b_{1,2}c_{2,2} - b_{1,2}c_{1,1} - b_{2,2}c_{1,2} \\
        t_{2,1} &= a_{1,1}d_{2,1} + a_{2,1}d_{1,1} + a_{2,1}d_{2,2} + a_{2,2}d_{2,1} - b_{1,1}c_{2,1} - b_{2,1}c_{1,1} - b_{2,1}c_{2,2} - b_{2,2}c_{2,1} \\
        t_{2,2} &= 2a_{2,2}d_{2,2} + a_{1,2}d_{2,1} + a_{2,1}d_{1,2} - 2b_{2,2}c_{2,2} - b_{1,2}c_{2,1} - b_{2,1}c_{1,2}
    \end{align*}
    We then define $\sdet_{2,2} A$ as the family of polynomials $\{\frac{1}{2}t_{1,1}, \frac{1}{2}t_{1,2}, \frac{1}{2}t_{2,1}, \frac{1}{2}t_{2,2} \}$.
\end{example}
  Using proof ideas similar to the proof of Theorem~\ref{thm:vnpac-complete}, we can show that this family is $\VNP$-complete in the commutative setting under $c$-reductions. First we show that the coefficient of every monomial for any $T_{i,j}$ (as per Definition~\ref{def:sdet-family-comm}) can be computed in polynomial time, thus putting the family of polynomials in $\VNP$ \cite{valiant1979completeness}. We then encode the Hamiltonian Cycle polynomial in commutative setting in $T_{a,b}$ for a fixed $a,b$ and then project to that particular position. 
\begin{theorem}\label{thm:vnpc-complete}
    $\sdet_{m,n}$ is $\VNP$-complete under $c$-reductions.
\end{theorem}
 We defer the proof to Appendix~\ref{app:sdet-comm-complete}.
 \begin{remark}
     In the usual setting, when we say $\{f_n\}$ is a family of polynomials, given an $n$ we get a single polynomial. In this setting, given a $n$ we get $m^2$ many polynomials and then we show that computing a fixed polynomial amongst them is hard under a specific algebra.
 \end{remark}
\section{Conclusion}
In this work, we studied the symmetrized determinants. We showed that analogous to ordinary determinants, the symmetrized determinants also follow the principal minor expansion property. This property allows the usual determinants of positive semi-definite matrices to be viewed as probability distributions with wide applications in machine learning, and exploring the symmetrized determinants through this lens is an exciting direction to pursue. We also showed that the symmetrized determinants are $\sharpP$-hard and $\vnpac$-complete for suitably chosen algebras of polynomial dimension. At the same time, it is known that for an algebra of dimension $r$, they can be computed in $O(n^{r+3})$. The appearance of $r$ in the exponent of $n$ seems unavoidable there, and proving this parametrized hardness (or finding an $\mathsf{FPT}$ algorithm for constant-dimensional algebras) for the symmetrized determinants is another interesting open problem.  

\newpage
\bibliography{main.bib}

\begin{thebibliography}{10}

\bibitem{arora2009computational}
Sanjeev Arora and Boaz Barak.
\newblock {\em Computational complexity: a modern approach}.
\newblock Cambridge University Press, 2009.

\bibitem{arvind2016noncommutative}
Vikraman Arvind, Pushkar~S Joglekar, and S~Raja.
\newblock Noncommutative valiant's classes: Structure and complete problems.
\newblock {\em ACM Transactions on Computation Theory (ToCT)}, 9(1):1--29, 2016.

\bibitem{arvind2010hardness}
Vikraman Arvind and Srikanth Srinivasan.
\newblock On the hardness of the noncommutative determinant.
\newblock In {\em Proceedings of the forty-second ACM symposium on Theory of computing}, pages 677--686, 2010.

\bibitem{barvinok1997computing}
Alexander Barvinok.
\newblock Computing mixed discriminants, mixed volumes, and permanents.
\newblock {\em Discrete \& Computational Geometry}, 18(2):205--237, 1997.

\bibitem{Barvinok2000NewPE}
Alexander~I. Barvinok.
\newblock New permanent estimators via non-commutative determinants.
\newblock {\em arXiv: Combinatorics}, 2000.
\newblock URL: \url{https://api.semanticscholar.org/CorpusID:117950899}.

\bibitem{bjorklund2012counting}
Andreas Bj{\"o}rklund.
\newblock Counting perfect matchings as fast as ryser.
\newblock In {\em Proceedings of the twenty-third annual acm-siam symposium on discrete algorithms}, pages 914--921. SIAM, 2012.

\bibitem{blaser2015noncommutativity}
Markus Bl{\"a}ser.
\newblock Noncommutativity makes determinants hard.
\newblock {\em Information and Computation}, 243:133--144, 2015.

\bibitem{burgisser2000completeness}
Peter B{\"u}rgisser.
\newblock {\em Completeness and reduction in algebraic complexity theory}, volume~7.
\newblock Springer Science \& Business Media, 2000.

\bibitem{chebotarev2006matrix}
Pavel Chebotarev and Elena Shamis.
\newblock Matrix-forest theorems.
\newblock {\em arXiv preprint math/0602575}, 2006.

\bibitem{chien2011almost}
Steve Chien, Prahladh Harsha, Alistair Sinclair, and Srikanth Srinivasan.
\newblock Almost settling the hardness of noncommutative determinant.
\newblock In {\em Proceedings of the forty-third annual ACM symposium on Theory of computing}, pages 499--508, 2011.

\bibitem{horn2012matrix}
Roger~A Horn and Charles~R Johnson.
\newblock {\em Matrix analysis}.
\newblock Cambridge university press, 2012.

\bibitem{hrubevs2010relationless}
Pavel Hrube{\v{s}}, Avi Wigderson, and Amir Yehudayoff.
\newblock Relationless completeness and separations.
\newblock In {\em 2010 IEEE 25th Annual Conference on Computational Complexity}, pages 280--290. IEEE, 2010.

\bibitem{warwickGCTnotes}
Christian Ikenmeyer.
\newblock Homework 9 - a first introduction to geometric complexity theory.
\newblock 2018.
\newblock URL: \url{https://dcs.warwick.ac.uk/~u2270030/teaching_sb/summer18/firstintrotogct/homework9.pdf}.

\bibitem{kulesza2012determinantal}
Alex Kulesza and Ben Taskar.
\newblock Determinantal point processes for machine learning.
\newblock {\em Foundations and Trends{\textregistered} in Machine Learning}, 5(2-3):123--286, 2012.

\bibitem{mahajan-vinay}
Meena Mahajan and V.~Vinay.
\newblock Determinant: Combinatorics, algorithms, and complexity.
\newblock {\em Chic. J. Theor. Comput. Sci.}, 1997.
\newblock URL: \url{http://cjtcs.cs.uchicago.edu/articles/1997/5/contents.html}.

\bibitem{moore2012approximating}
Cristopher Moore and Alexander Russell.
\newblock Approximating the permanent via nonabelian determinants.
\newblock {\em SIAM Journal on Computing}, 41(2):332--355, 2012.

\bibitem{nisan1991lower}
Noam Nisan.
\newblock Lower bounds for non-commutative computation.
\newblock In {\em Proceedings of the twenty-third annual ACM symposium on Theory of computing}, pages 410--418, 1991.

\bibitem{ramprasadsurvey}
Ramprasad Saptharishi.
\newblock A survey of lower bounds in arithmetic circuit complexity.
\newblock {\em Github survey}, 95, 2015.

\bibitem{valiant1979completeness}
Leslie~G Valiant.
\newblock Completeness classes in algebra.
\newblock In {\em Proceedings of the eleventh annual ACM symposium on Theory of computing}, pages 249--261, 1979.

\end{thebibliography}
\newpage
\appendix
\section{$\vnpac$-completeness of $\hcac$}\label{app:hcac-completeness}

In this section, we will show that $\hcac$, as defined in Definition~\ref{def:hcac-poly} is $\vnpac$ complete. Hrube\v{s} et. al. \cite[Corollary 12]{hrubevs2010relationless} showed that any polynomial in $\vnpac$ can also be simulated by an exponential boolean sum over polynomial sized \emph{formulas}. Hence, to establish that $\hcac$ is $\vnpac$-complete, we need to show two things: (1) Any formula can be simulated by a ``small'' $\hcac$ polynomial, \emph{and} (2) Boolean sums can be simulated by ``small'' $\hcac$ polynomials. It is also known \cite[Lemma 2.15]{burgisser2000completeness} that every polynomial sized formula $f$ can be represented as a polynomial sized DAG $G$ with identifiable $s,t$ nodes such that the sum of weights of all $s$ to $t$ paths of $G$ equals $f$. It is easy to verify that also holds in non-commutative settings. Hence, to establish (1) we need to show that such DAGs can be simulated by a polynomial sized $\hcac$ polynomial. We show this in the following lemma:
\begin{lemma}[$\hcac$ simulates Formulas]\label{lem:hcac-sim-formulas}
    Any arithmetic formula can be simulated by the $\hcac$ polynomial with only a polynomial blowup in size.
\end{lemma}
\begin{proof}
    Using \cite[Lemma 2.15]{burgisser2000completeness}, let $\hat G=(V,E)$ be the DAG representing the arithmetic formula $f$, i.e, the weighted sum of all $s\rightsquigarrow t$ paths in $G$ is equal to $f$. It is easy to see that there exists an equivalent graph $G$ that can be partitioned in layers $L_0,...,L_d$ where $L_0 = \{s\}, L_d  = \{t\}, L_i = \{v_{i,1},\ldots, v_{i,k_i}\}$ such that each edge goes from some vertex in layer $L_i$ to layer $L_{i+1}$ and the sum of all $s \rightsquigarrow t$ is same in both $\hat G$ and $G$. Further, $s$ (resp. $t$) have indegree (resp. outdegree) as $0$. Construct $G'=(V',E')$ such that $V'=V$ and we define $E'$ in the following way:\begin{itemize}
        \item Add edges from $v_{i,j}$ to $v_{i, (j+1) \mod k_i}$ for each $j \in [k_i]$ for all $i \in \{1,2,\ldots,d-1\}$ with edge weight $1$. 
        \item For edges of the form $(v_{d-1,j},t)$ with edge weight $w$ in $E$, add them to $E'$ for all $j \in [k_{d-1}]$.
        \item For all edges from $s$ to $v_{1,j}$ with edge weight $w$ in $E$, add the edge $s$ to $v_{1,(j+1) \mod k_1}$ with weight $w$ in $E'$.
        \item Let $i\in \{1,\ldots,d-2\}$. For all edges of the form $(v_{i,j},v_{i+1,l})$ with edge weight $w$ in $E$, add the edge $(v_{i,j},v_{i+1,(l+1)\mod k_{i+1}})$ with weight $w$ to $E'$.
        \item Add an edge $(t,s)$ with weight $1$ to $E'$. 
    \end{itemize} 
    We will show that $s\rightsquigarrow t$ paths in $G$ with weight $w$ have a $1-1$ correspondence with Hamiltonian cycles of weight $w$ in $G'$. Further, the nodes are traversed in the same order.
    
Consider the path $s \rightarrow v_{1,a_1 \mod k_1} \rightarrow \ldots v_{d-1,a_{d-1} \mod k_{d-1}} \rightarrow t$. Then we have the following cycle in $G'$: 
\begin{align*}
    s \\ \rightarrow \big( v_{1,a_1 + 1 \mod k_1} \rightarrow v_{1,a_1 + 2 \mod k_1} \ldots \rightarrow v_{1,a_1 \mod k_1}  \big) \\ \rightarrow \big( v_{2,a_2 + 1 \mod k_2} \rightarrow v_{1,a_2 + 2 \mod k_2} \ldots \rightarrow v_{1,a_2 \mod k_2}  \big) \\ \rightarrow \ldots  \\ \rightarrow \big( v_{d-1,a_{d-1} + 1 \mod k_{d-1}}  \ldots \rightarrow v_{d-1,a_{d-1} \mod k_{d-1}}  \big) \\ \rightarrow t \rightarrow s
\end{align*}
Clearly this cycle visits each vertex exactly once in $G'$ and is Hamiltonian. Further, the weight is the same and in the same order. Now consider any Hamiltonian cycle in $G'$ beginning from $s$. It must visit each vertex exactly once. Due to the graph being layered and acyclic, once we leave a particular layer there is no way to return to it. So we must traverse all the nodes in it. Thus, if we go from $s$ to $v_{1,a_1 \mod k_1}$ we must travel to $v_{1,(a_1 - 1) \mod k_1}$ before going to the second layer to say $v_{2,a_2 \mod k_2}$. But this means we must have edges $s$ to $v_{1,(a_1-1) \mod k_1}$ to $v_{2, (a_2-1) \mod k_2}$ in $G$. Arguing inductively, every Hamiltonian cycle in $G'$ has a corresponding ordered path of the same weight in $G$. Since $V'=V$ and the only new edges in $E'$ are $k_i$ per layer $i$, size of $G'$ only increases polynomially in the size of input graph $\hat G$.
\end{proof}

We now have to prove that Boolean sums can be efficiently simulated by the $\hcac$ polynomial. To prove this, we will take the help of ``Rosettes'' and ``Glue'' graphs (see \cite{burgisser2000completeness,ramprasadsurvey}) conventionally used to prove the $\VNP$-completeness of the \emph{Permanent} polynomial family. However, we will need to make some modifications to adapt them to $\hcac$ polynomial. We will base our constructions on those that appear in \cite{warwickGCTnotes}.

Formally, we have $f(x_1,\ldots,x_n) = \sum_{\bfa \in \{0,1\}^m} g(x_1,\ldots,x_n,a_1,\ldots,a_m)$. By Lemma~\ref{lem:hcac-sim-formulas}, we have that there is a graph $G$ with fixed start and sink node $s$ and $t$ such that the weighted sum over all Hamiltonian cycles beginning from $s$ is equal to $g(\bfx, \bfy)$. We now want to compute $\sum_{\bfa} g(\bfx, \bfa)$.

We first look at the simpler case where there exists $y \in \bfy$ such that there is only one edge labeled with $y$ in $G'$. When we take the boolean sum over $y$, it would mean we count the cycles which included $y$ once and those that didn't include $y$ twice. By our construction in the proof of Lemma~\ref{lem:hcac-sim-formulas}, an edge with weight label other than $1$ can only be present between different layers, as all edges within the same layer have weight $1$. Hence, we can assume that the edge $e_y$ with weight $y$ was between some node $v_a$ in layer $i$ and $v_b$ in layer $i+1$. Then, we can modify our graph $G$ by replacing weight of edge between $v_a$ and $v_b$ with $1$ and for all other edges between layers $i$ and $i+1$, we multiply the weight by a constant factor $2$. Thus, any Hamiltonian cycle in $G$ beginning at $s$ which took the edge with weight $y$ is counted once, while all other Hamiltonian cycles have their weight doubled, as each cycle must cross from layer $i$ to layer $i+1$ at some point. Further, the order of traversal of the cycles remains the same and hence the order of multiplication is never altered. 

Without loss of generality, we can now consider that all $y \in \bfy$ appear more than once as edge weights. This case is considerably trickier, and we will now use auxiliary graphs -- the ``Rosette'' and the ``Glue'' graph -- to modify our original graph $G$ such that it can simulate Boolean sums. We begin by describing these gadgets based on \cite{warwickGCTnotes}. Intuitively, the aim of $R_g$ is to ensure that if we avoid a particular subset of edges (called the \emph{indicator} edges) in $R_g$, then we have two $s \rightsquigarrow t$ paths, while if we take any edge from that subset there is only one such path. An example of the Rosette Graph ($R_g^5 = (V_r,E_r)$) with $5$ indicator edges is shown in Figure~\ref{fig:rosette}. Here each edge has weight $1$. The dashed edges represent the \emph{indicator} edges. There are \emph{exactly} two Hamiltonian paths from $s$ to $t$ that don't include any indicator edges. These paths are $p_1 \coloneq s \rightarrow 1_i \rightarrow 1_o \rightarrow b_1 \rightarrow \ldots \rightarrow 5_i \rightarrow 5_o \rightarrow t$, and $p_2 \coloneq s \rightarrow 5_i \rightarrow 5_o \rightarrow b_4 \rightarrow \ldots \rightarrow 1_i \rightarrow 1_o \rightarrow t$. Further, by construction, there is \emph{exactly} one Hamiltonian path from $s$ to $t$ including any non-empty subset of indicator edges. We can draw an analogous Rosette Graph $R_g^i$ for all $i\geq 2$, where $i$ is the number of indicator edges.
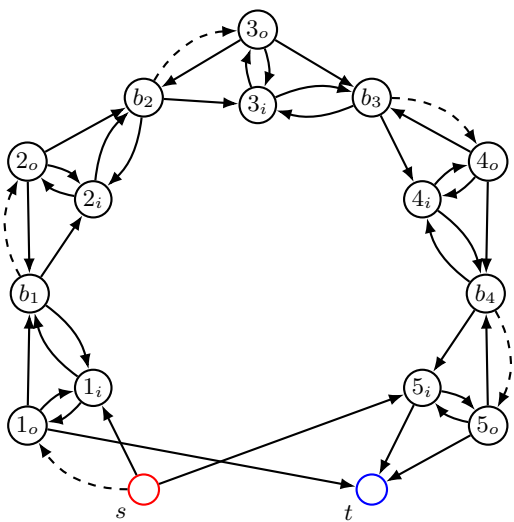
\begin{figure}[h]
    \centering

\begin{tikzpicture}[
    >={Latex[length=2mm]},
    every node/.style={circle,draw,inner sep=1.2pt,minimum size=4mm,font=\small},
    every path/.style={->,thick}
]

\node[label=below left:$s$, color=red] (s) at ({-3*sin(30)},{-3*cos(30)}) {};
\node[label=below left:$t$, color=blue] (t) at ({3*sin(30)},{-3*cos(30)}) {};

\node[] (b1) at ({-3},{0}) {$b_1$};
\node[] (b2) at ({-3*sin(30)},{3*cos(30)}) {$b_2$};
\node[] (b3) at ({3*sin(30)},{3*cos(30)}) {$b_3$};
\node[] (b4) at ({3},{0}) {$b_4$};

\node[] (1i) at ({-2.5*sin(60)},{-2.5*cos(60)}) {$1_i$};
\node[] (1o) at ({-3.5*sin(60)},{-3.5*cos(60)}) {$1_o$};
\node[] (2i) at ({-2.5*sin(60)},{2.5*cos(60)}) {$2_i$};
\node[] (2o) at ({-3.5*sin(60)},{3.5*cos(60)}) {$2_o$};
\node[] (3i) at ({0},{2.5}) {$3_i$};
\node[] (3o) at ({0},{3.5}) {$3_o$};
\node[] (4i) at ({2.5*sin(60)},{2.5*cos(60)}) {$4_i$};
\node[] (4o) at ({3.5*sin(60)},{3.5*cos(60)}) {$4_o$};
\node[] (5i) at ({2.5*sin(60)},{-2.5*cos(60)}) {$5_i$};
\node[] (5o) at ({3.5*sin(60)},{-3.5*cos(60)}) {$5_o$};

 \draw[->] (1i) to[bend left=20] (1o);
 \draw[->] (1o) to[bend left=20] (1i);
 \draw[->] (2i) to[bend left=20] (2o);
 \draw[->] (2o) to[bend left=20] (2i);
 \draw[->] (3i) to[bend left=20] (3o);
 \draw[->] (3o) to[bend left=20] (3i);
 \draw[->] (4i) to[bend left=20] (4o);
 \draw[->] (4o) to[bend left=20] (4i);
 \draw[->] (5i) to[bend left=20] (5o);
 \draw[->] (5o) to[bend left=20] (5i);

 \draw[->] (s) -- (1i);
 \draw[->] (b1) -- (2i);
 \draw[->] (b2) -- (3i);
 \draw[->] (b3) -- (4i);
 \draw[->] (b4) -- (5i);

 \draw[->] (1i)  to[bend left=20] (b1);
 \draw[->] (1o) -- (b1);
 \draw[->] (2i)  to[bend left=20] (b2);
 \draw[->] (2o) -- (b2);
 \draw[->] (3i)  to[bend left=20] (b3);
 \draw[->] (3o) -- (b3);
 \draw[->] (4i)  to[bend left=20] (b4);
 \draw[->] (4o) -- (b4);
 \draw[->] (5i) -- (t);
 \draw[->] (5o) -- (t);

 \draw[->] (s) -- (5i);
 \draw[->] (b1)  to[bend left=20] (1i);
 \draw[->] (b2)  to[bend left=20] (2i);
 \draw[->] (b3)  to[bend left=20] (3i);
 \draw[->] (b4)  to[bend left=20] (4i);

 \draw[->] (1o) -- (t);
 \draw[->] (2o) -- (b1);
 \draw[->] (3o) -- (b2);
 \draw[->] (4o) -- (b3);
 \draw[->] (5o) -- (b4);

\draw[->, dashed] (s)  to[bend left=28] (1o);
\draw[->, dashed] (b1)  to[bend left=28] (2o);
\draw[->, dashed] (b2)  to[bend left=28] (3o);
\draw[->, dashed] (b3)  to[bend left=28] (4o);
\draw[->, dashed] (b4)  to[bend left=28] (5o);

\end{tikzpicture}
    \caption{Rosette Graph with $5$ connector edges.}
    \label{fig:rosette}
\end{figure}

The Glue graph ($G_g = (V_g,E_g)$) is shown in Figure~\ref{fig:glue}. Here all the edge weights are $1$. It is used to ``glue'' together two edges $\eps=(u,v)$ and $\eps'=(u',v')$ with weights $w,w'$ appearing in any graph $H=(V,E)$ with start and sink nodes $\hat s,\hat t$. We define a new graph $H' = (V',E')$ with the node set of $G_g \cup H$. We first set $V' = V \cup V_g$. We now define $E'$. First, we add $E_g$ to $E'$. Now, let $E_- = E \setminus \{ (u,v), (u',v'), (\hat t, \hat s) \}$. Then, we add $E_-$ to $E'$. Further, we add edges of weight $1$ between 
$(t, \hat s) \text{ and } (\hat t,s) $ in $E'$. Now, suppose there was a Hamiltonian Cycle $C\coloneq \hat s \rightarrow x_1 \rightarrow \ldots \rightarrow x_m \rightarrow \hat t \rightarrow \hat s$ in $H$ that didn't take the edges $\eps$ or $\eps'$. Then, $C'\coloneq \hat s \rightarrow x_1 \rightarrow \ldots \rightarrow x_m \rightarrow \hat t \rightarrow s \rightarrow c \rightarrow b \rightarrow a \rightarrow d \rightarrow e \rightarrow f \rightarrow t \rightarrow \hat s$ is a Hamiltonian cycle in $H'$ of the same weight and is easy to verify that this is the only Hamiltonian Cycle that includes $C$. Now suppose, there was a Hamiltonian Cycle $C\coloneq \hat s \rightarrow \ldots \rightarrow x \rightarrow u \rightarrow v \rightarrow \ldots u'  \rightarrow v' \ldots \rightarrow  x_m \rightarrow \hat t \rightarrow \hat s$ in $H$ that includes both $\eps, \eps'$. Then, in $H'$, $C'\coloneq \hat s \rightarrow \ldots \rightarrow x \rightarrow u \rightarrow a \rightarrow b \rightarrow c \rightarrow v \rightarrow \ldots u' \rightarrow f \rightarrow e \rightarrow d  \rightarrow v' \ldots \rightarrow  x_m \rightarrow \hat t \rightarrow  s \rightarrow t \rightarrow \hat s$ is a Hamiltonian Cycle in $H'$ and in fact the only one corresponding to $C$. However, if we take a Hamiltonian cycle in $H$ that includes one of $\eps, \eps'$ but not the other, by the construction of $G_g$ it is impossible to get a corresponding Hamiltonian cycle in $H'$. Hence, $G_g$ essentially glues the edges $\eps, \eps'$, that is, either none of them will be considered in the Hamiltonian Cycle, or both will be considered. Further, when they are not taken the weight remains the same, while when they are both taken the edge weights of $\eps,\eps'$ are replaced by $1$.

\begin{figure}[h]
    \centering
  
\begin{tikzpicture}[
    >=Latex,
    every node/.style={circle,draw,inner sep=1.2pt,minimum size=4mm,font=\small},
    every path/.style={->,thick}
]

\node[label=above:$u$]  (u)  at (0,2.2) {};
\node[label=above:$v$]  (v)  at (3,2.15) {};
\node[label=right:$s$, color=red]  (s)  at (4.5,1.1) {};
\node[label=right:$t$, color=blue]  (t)  at (4.5,-0.7) {};
\node[label=below:$u'$] (up) at (3,-2.2) {};
\node[label=below:$v'$] (vp) at (0,-2.25) {};

\node (a) at (0,1.1) {$a$};
\node (b) at (1.5,1.1) {$b$};
\node (c) at (3,1.1) {$c$};

\node (d) at (0,-0.7) {$d$};
\node (e) at (1.5,-0.7) {$e$};
\node (f) at (3,-0.7) {$f$};

\draw (u)  -- (a);
\draw (c)  -- (v);
\draw (s)  -- (c);
\draw (f)  -- (t);
\draw (d)  -- (vp);
\draw (up) -- (f);
\draw (s) -- (t);

\draw (a) to[bend left=18] (b);
\draw (b) to[bend left=18] (a);

\draw (b) to[bend left=18] (c);
\draw (c) to[bend left=18] (b);

\draw (d) to[bend left=18] (e);
\draw (e) to[bend left=18] (d);

\draw (e) to[bend left=18] (f);
\draw (f) to[bend left=18] (e);

\draw (a) -- (d);
\draw (c) -- (f);

\end{tikzpicture}
    \caption{Glue graph for edges $(u,v)$ and $(u',v')$.}
    \label{fig:glue}
\end{figure}
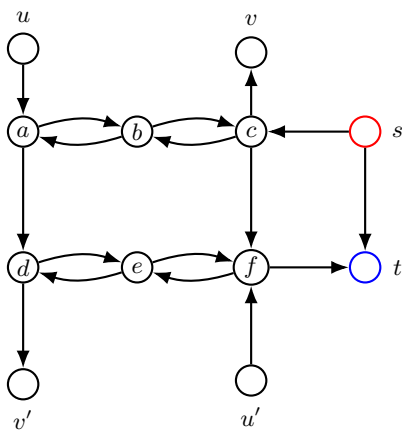

We now describe how these gadgets help us to simulate Boolean sums for $y \in \bfy$ when all these $y$ appear multiple times as edge weights in $G = (V,E)$. Let $G$ have start and sink nodes $s_0,t_0$. First choose some $y \in \bfy$. Suppose $y$ occurs as edge weight $i$ times in $G$. We now look at the Rosette graph $R_g = (V_r,E_r)$ with $i$ indicator edges and construct $G_1 = (V_1,E_1)$ in the following way: $V_1 = V \cup V_r$. For $E_1$, we add all edges of $E_r$ in $E_1$. We also add all edges in $E$ except the edge from $t_0$ to $s_0$ to $E_1$. Finally, we add an edge from $t_0$ to $s$ of the Rosette Graph with edge weight as $1$. Intuitively, this ensures that the contributions from the paths in Rosette Graph are multiplied with the paths in the original graph $G$. Then, as described above, we can use the glue graph to glue every indicator edge with some occurrence of $y$ in $G_1$ to get a Hamiltonian Cycle of the same weight \emph{and} order of multiplication preserved. Now, notice that the size $R_g$ with $i$ indicator edges has $3i+1$ nodes and hence $O(i^2)$ edges. So total size of $R_g$ is $O(i^2)$. Similarly, any glue graph has constantly many nodes and edges, and thus size of $G_g$ is $O(1)$. Let the original graph $G$ have $n$ nodes. Then, it has at most $n^2$ edges. Thus, the sum of occurrences of each $y \in \bfy$ in $G$ as edge weight is $\leq n^2$. Thus, adding Rosette Graphs for each of these $y$ only leads to at most $O(n^4)$ overhead to create $G_1$. Then adding at most $n^2$ glue graphs increases size by at most by $O(n^2)$. Hence, the final graph $G_{final}$ has size polynomial in the size of the input graph, and it allows us to simulate Boolean sums using the $\hcac$ polynomial. 
\begin{proposition}
    Boolean sums over formulas can be simulated by $\hcac$ with only a polynomial overhead.
\end{proposition}
Also, observe that in all our reductions, we only added edges of weight $1$ which were multiplied with the original edges. Hence, the final poylnomial $f(x_1,...,x_n)$ is recovered exactly from the $\hcac$ polynomial over $G_{final}$. Hence, using this alongside Lemma~\ref{lem:hcac-sim-formulas} shows that $\hcac$ is $\VNP$-complete under $p$-projections, as defined by Valiant in \cite{valiant1979completeness}.
\begin{remark}
    $p$-projection is a much stronger form of reduction as compared to $c$-reduction, see \cite{burgisser2000completeness} for details.
\end{remark}

\section{Proof of Theorem~\ref{thm:vnpc-complete}}\label{app:sdet-comm-complete}
We first note that $\hcac$ polynomial as defined in Definition~\ref{def:hcac-poly} is also $\VNP$-complete in the commutative setting. This is because, using commutativity, we can write it as $\hcac = \sum_{\sigma \in C_n} x_{1,\sigma(1)}x_{2, \sigma(2)} \cdots x_{n, \sigma(n)}$ which is known to be $\VNP$-complete \cite{valiant1979completeness}.

We now show that the family of polynomials $T_{i,j}$ (as defined in Definition~\ref{def:sdet-family-comm} is in $\VNP$ in the commutative setting. First see that using (\ref{eq:sdet}) in Definition~\ref{def:sdet-family-comm}, we get 
\begin{align}\label{eq:tkl-sdetfam}
    T_{k,l} = \frac{1}{n!} \sum_{(\sigma, \tau) \in \Sym^2_n} \sum_{x_1=1}^m \ldots \sum_{x_{n-1}=1}^m \sgn \sigma \sgn \tau \  m^{\sigma(1),\tau(1)}_{k,x_1}m^{\sigma(2),\tau(2)}_{x_1,x_2}\cdots m^{\sigma(n-1),\tau(n-1)}_{x_{n-2},x_{n-1}}m^{\sigma(n),\tau(n)}_{x_{n-1},l}
\end{align}
 Given any monomial $\mathfrak{m} = 
 \prod_{\substack{i,j \in [n]\\a,b \in [m] }} m^{i,j}_{a,b}$, we can compute the coefficient of $\mathfrak{m}$ in $T_{k,l}$ in polynomial time. First it is easy to see that every non-zero monomial must have terms from different $m^{i,j}$, and further for a fixed $i,j$, the powers of $m^{i,j}_{k,l}$ is at most $1$ for any $k,l \in [m]$. Hence, we now concern ourselves with only these non-zero monomials. Suppose we want to find the coefficient of the monomial $\mathfrak{m} = m^{a_1,b_1}_{x_1,y_1}m^{a_2,b_2}_{x_2,y_2} \cdots m^{a_n,b_n}_{x_n,y_n}$ in $T_{k,l}$. Then, by above we know that $\{a_1,\ldots,a_n\} = \{b_1,\ldots,b_n\} = [n]$, and $x_i,y_i \in [m]$ for all $i$. By, (\ref{eq:tkl-sdetfam}), we know that there must exist $r \neq s \in [n]$ such that $k=x_r$ and $l=y_s$ for the coefficient of $\mathfrak{m}$ to be non-zero in $T_{k,l}$. Further, by commutativity of $m^{i,j}_{a,b}$ all permutations such that $\sigma(i) = a_k, \tau(i) = b_k$ can give rise to $\mathfrak{m}$. Now suppose there are $c_1$ many indexes $r$ such that $k = x_r$, and $c_2$ many indexes $s$ such that $l = y_s$. By commutativity, we can rearrange $\mathfrak{m}$ as \[
 \mathfrak{m} = m^{a_r,b_r}_{k,y_r}\cdots m^{a_s,b_s}_{x_s,l}
 \]
 This fixes $\sigma(1) = a_r$ and $\tau(1) = b_r$ and $\sigma(n) = a_s$ and $\tau(s) = b_s$. While for the other $n-2$ indexes, we have $(n-2)!$ choices for $\sigma$, which also fixes $\tau$ as $\sigma(i) = a_k$ implies $\tau(i) = b_k$. But this simply means $\alpha = \tau \circ \sigma^{-1}$ is the permutation that sends $a_i$ to $b_i$ for each $i$. Hence, $\sgn \sigma \sgn \tau = \sgn \alpha$ is fixed for this monomial. Further, there are $K = c_1c_2 - c_3$ many ways ways to rearrange $\mathfrak{m}$. The $c_3$ here counts the number of terms which have $x_r=k$ \emph{and} $y_s=l$, and we subtract this to avoid double counting. Thus, incorporating all this, we get that the coefficient of $\mathfrak{m} = m^{a_1,b_1}_{x_1,y_1}m^{a_2,b_2}_{x_2,y_2} \cdots m^{a_n,b_n}_{x_n,y_n}$ in $T_{k,l}$ is $\frac{K (n-2)!}{n!} \sgn \alpha$. Since $K,\alpha$ can be computed in polynomial time given $\mathfrak{m}$, we get each coefficient in polynomial time.

Finally, to show $\VNP$-completeness, we adapt the proof of Theorem~\ref{thm:vnpac-complete}. Given a directed graph $G$, we set $M_{i,j}$ as before, with the only difference being that $x_{i,j}$ are now commutative variables. More concretely, if there is an edge from $1$ to $j$, we set $A_{i,j} = x_{1,j} f \otimes e_{1,j}$. We can similarly set all entries, similar to proof of Theorem~\ref{thm:vnpac-complete}. Since $f$ was $2 \times 2$ and $e_{i,j}$ was $n \times n$, $\alg = \mat(\Q,2n,2n)$. Then using the same argument, in the end we get $\sdet A = \frac{(-1)^n}{n!} \hcac (f \otimes e_{1,1}) $. Hence, $\sdet_{2n,n} M = \{T_{1,1}, \ldots, T_{2n,2n}\}$ is such that $T_{i,j} = 0$ for all but one pair $(i,j)$. In particular, $T_{1,n+1} = \frac{(-1)^n}{n!} \hcac$. Thus, computing it is $\VNP$-hard. But under $c$-reductions, we can take projections of $\sdet M$. Concretely, $n!(-1)^ne_1^T (\sdet M )e_{n+1} = \hcac$ where $e_i$ are the unit column vectors in $i$-th direction. Thus, we can conclude that $\sdet_{m,n}$ is $\VNP$-complete under $c$-reductions.

\end{document}